\documentclass[conference]{IEEEtran}
\IEEEoverridecommandlockouts

\usepackage{color}

\usepackage{ifthen}
\usepackage{substr}
\makeatletter
\@ifundefined{paperVersion}{\newcommand{\paperVersion}{paper}}{} 
\makeatother

\usepackage[noadjust]{cite}
\usepackage{marcpaper}
\usepackage{xspace}
\usepackage{color}
\usepackage{array}

\usepackage[noend]{algpseudocode}

\usepackage{algorithm}
\usepackage{eqparbox}
\usepackage{pifont}

\newcommand{\ignore}[1]{}

\newcommand{\hiddencomment}[1]{}
\newcommand{\oldcomment}[1]{}

\newcommand{\var}[1]{\ensuremath{\mathsf{#1}}}

\newcommand{\condSafe}{\ensuremath{\mathsf{CondSafe}}\xspace}
\newcommand{\checkSafe}{\ensuremath{\mathsf{CheckSafe}}\xspace}

\newcommand{\sat}{SAT\xspace}
\newcommand{\maxsat}{Max-SAT\xspace}
\newcommand{\smt}{SMT\xspace}
\newcommand{\maxsmt}{Max-SMT\xspace}

\newcommand{\MAXSAT}{Max-SAT}

\newcommand{\MAXSMT}{Max-SMT}

\newcommand{\limplies}{\Rightarrow}

\def\defemb#1#2{\expandafter\def\csname #1\endcsname
                   {\relax\ifmmode #2\else\hbox{$#2$}\fi}}
\defemb{cA}{{\cal A}}
\defemb{cB}{{\cal B}}
\defemb{cC}{{\cal C}}
\defemb{cE}{{\cal E}}
\defemb{cL}{{\cal L}}
\defemb{cN}{{\cal N}}
\defemb{cP}{{\cal P}}
\defemb{cS}{{\cal S}}
\defemb{cT}{{\cal T}}
\defemb{cU}{{\cal U}}
\defemb{cI}{{\cal I}}
\defemb{cD}{{\cal D}}
\defemb{cS}{{\cal S}}
\defemb{cV}{{\cal V}}
\defemb{cQ}{{\cal Q}}
\defemb{cR}{{\cal R}}

\newcommand{\entries}{\ensuremath{\cE}\xspace}

\newcommand{\formula}       {\mathbb{F}}

\newcommand{\initiation}    {\mathbb{I}}

\newcommand{\consecution}   {\mathbb{C}}
\newcommand{\safety}        {\mathbb{S}}

\newcommand{\define}        {\stackrel{\mathit{def}}{=}}

\newcommand{\true}               {\top}

\newcommand{\weight}               {\omega}

\newcommand{\weightInitiation}     {\weight_{\initiation}}

\newcommand{\Var}               {p}

\newcommand{\VarInitiation}[1]  {\Var_{\initiation_{#1}}}

\newcommand{\vars}{\cV}

\DeclareMathOperator{\QF}{\mathcal{F}}

\newtheorem{definition}{Definition}
\newtheorem{theorem}{Theorem}
\newtheorem{example}{Example}

\begin{document}

\title{Compositional Safety Verification with \maxsmt}

\author{\IEEEauthorblockN{
    Marc Brockschmidt	          \IEEEauthorrefmark{1},
    Daniel Larraz		  \IEEEauthorrefmark{2},
    Albert Oliveras		  \IEEEauthorrefmark{2}, 
    Enric Rodr\'\i guez-Carbonell \IEEEauthorrefmark{2} and
    Albert Rubio                  \IEEEauthorrefmark{2}}
  \IEEEauthorblockA{\IEEEauthorrefmark{1}Microsoft Research, Cambridge}
  \IEEEauthorblockA{\IEEEauthorrefmark{2}Universitat Polit\`ecnica de Catalunya}  

  \thanks{
    This work has been supported by Spanish MINECO under the
    grant TIN2013-45732-C4-3-P (project DAMAS) and the FPI grant
    BES-2011-044621 (Larraz).
  } }

\maketitle

\begin{abstract}
We present an automated compositional program verification technique for
safety properties based on \emph{conditional inductive invariants}. For a given
program part (e.g., a single loop) and a postcondition $\varphi$,
we show how to, using a Max-SMT solver,
an inductive invariant together with a
precondition
can be synthesized
so that the precondition ensures the validity of the invariant and
that the invariant implies $\varphi$.
From this, we build a bottom-up program verification framework that propagates
preconditions of small program parts as postconditions for preceding program
parts. The method recovers from failures to prove the validity of a precondition,
using the obtained intermediate results to restrict the search space for
further proof attempts.

As only small program parts need to be handled at a time, our method is
scalable and distributable. The derived conditions can be viewed as implicit
contracts between different parts of the program, and thus enable an incremental
program analysis.
\end{abstract}


\section{Introduction}

To have impact on everyday software development, a verification engine needs
to be able to process the millions of lines of code often encountered in
mature software projects.
At the same time, the analysis should be repeated every time developers commit a
change, and should report feedback in the course of minutes, so that fixes
can be applied promptly.
Consequently, a central theme in recent research on automated program
verification has been \emph{scalability}.
As a natural solution to this problem, \emph{compositional}
program analyses~\cite{Godefroid10,Calcagno11,Li13} have been
proposed. They analyze program parts \mbox{(semi-)independently} and then
combine the results to obtain a whole-program proof.

For this, a compositional analysis has to predict likely intermediate assertions
that allow us to break whole-program reasoning into many instances of local
reasoning.
This strategy simplifies the individual reasoning steps and allows
distributing the analysis~\cite{Albarghouthi12b}.
The disadvantage of compositional analyses has traditionally been their
precision: local analyses must blindly choose the intermediate assertions.
While in some domains ({\em e.g.} heap) some heuristics have been
found~\cite{Calcagno11}, effective strategies for guessing and/or refining
useful intermediate assertions or summaries in arithmetic domains remains an open
problem.

In this paper we introduce a new method for predicting and refining
intermediate arithmetic assertions for compositional reasoning about
sequential programs. A key component in our approach is \maxsmt
solving. \maxsmt solvers can deal with hard and soft constraints,
where hard constraints are mandatory, and soft constraints are those
that we would like to hold, but are not required to. Hard constraints
express what is needed for the soundness of our analysis, while soft
ones favor the solutions that are more useful for our technique. More
precisely, we use \maxsmt to iteratively infer \emph{conditional
  inductive invariants}\footnote{This concept was previously
  introduced with the name ``quasi-invariant'' in \cite{Larraz14} in
  the different context of proving program non-termination.}, which
prove the validity of a property, given that a precondition holds.
Hence, if the precondition holds, the program is proved safe.
Otherwise, thanks to a novel program transformation technique we call
\emph{narrowing}\footnote{This narrowing is
  inspired by the narrowing in term rewrite systems, and is unrelated to the
  notion with the same name used in abstract interpretation.},
we exploit the failing conditional invariants to
focus on what is missing in the safety proof of the program. Then
new conditional invariants are sought, and the process is repeated until
the safety proof is finally completed. Based on this, we
introduce a new bottom-up program analysis procedure that infers
conditional invariants in a goal-directed manner, starting from a
property that we wish to prove for the program. Our approach makes
distributing analysis tasks as simple as in other bottom-up analyses,
but also enjoys the precision of CEGAR-based provers.


\section{Illustration of the Method}
In this section, we illustrate the core concepts of our approach by
using some small examples. We will give the formal definition of the
used methods in \rSC{sect:safety}.

We handle programs by considering one strongly connected component (SCC)
$\cC$ of the control-flow graph at a time, together with the sequential parts
of the program leading to $\cC$, either from initial states or other
SCCs.

Instead of program invariants, for each SCC we synthesize
\emph{conditional inductive invariants}. These are inductive
properties such that they may not always hold whenever the SCC is
reached, but once they hold, then they are always satisfied.

\paragraph{Conditional Inductive Invariants}
\begin{wrapfigure}[7]{r}{2.45cm}
  \vspace{-1.5ex}
  \begin{tabular}{l}
   \textbf{while} $\var{i} > 0$ \textbf{do} \\
   \quad$\var{x} := \var{x} + 5$;\\
   \quad$\var{i} := \var{i} - 1$;\\
   \textbf{done}\\
   \textbf{assert}($\var{x} \geq 0$);
  \end{tabular}
  \vspace*{-2ex}
  \caption{\label{fig:ex1}}
\end{wrapfigure}

As an example, consider the program snippet in \rF{fig:ex1}, where we
do not assume any knowledge about the rest of the program. To prove
the assertion, we need an inductive property $\QQ$ for the loop
\ignore{preceding it,} such that $\QQ$ together with the negation of the loop
condition $\var{i} > 0$ implies the assertion. Using our
constraint-solving based method \condSafe
(cf. \rSC{sect:Conditional-Safety}), we find $\QQ_1 = \var{x} + 5 \cdot
\var{i} \geq 0$. 
The property $\QQ_1$ can be seen as a \emph{precondition} at the loop entry
for the validity of the assertion.

\paragraph{Combining Conditional Inductive Invariants}
Once we have found a
conditional inductive invariant for an SCC, we use the generated
preconditions as postconditions for its preceding SCCs in the program.

\begin{wrapfigure}[6]{l}{2.45cm}
  \vspace{-1.5ex}
  \begin{tabular}{l}
    \textbf{while} $\var{j} > 0$ \textbf{do} \\
    \quad$\var{j} := \var{j} - 1$;\\
    \quad$\var{i} := \var{i} + 1$;\\
    \textbf{done}\\
  \end{tabular}
  \caption{\label{fig:ex2}}
\end{wrapfigure}

As an example, assume that the loop from \rF{fig:ex1} is directly
preceded by the loop in \rF{fig:ex2}. We now use the precondition
$\QQ_1$ we obtained earlier as input to our conditional invariant
synthesis method, similarly to the assertion in \rF{fig:ex1}.
Thus, we now look for an inductive property $\QQ_2$ that, together with $\lnot
(\var{j} > 0)$, implies $\QQ_1$.
In this case we obtain the conditional invariant $\QQ_2 = \var{j} \geq 0 \land
\var{x} + 5 \cdot (\var{i} + \var{j}) \geq 0$ for the loop. As with
$\QQ_1$, now we can see $\QQ_2$ as a precondition at the loop entry,
and propagate $\QQ_2$ up to the preceding SCCs in the program.

\paragraph{Recovering from Failures}
When we cannot prove that a precondition always holds, we try to recover
and find an alternative precondition. In this process, we make
use of the results obtained so far, and \emph{narrow}
the program using our intermediate results. As an example, consider the loop in
\rF{fig:ex3}.

\begin{wrapfigure}[7]{r}{3.3cm}
  \vspace{-1.5ex}
  \begin{tabular}{l}
    \textbf{while} $\mathit{unknown()}$ \textbf{do} \\
    \quad\textbf{assert}($\var{x} \neq \var{y}$);\\
    \quad$\var{x} := \var{x} + 1$;\\
    \quad$\var{y} := \var{y} + 1$;\\
    \textbf{done}\\
  \end{tabular}
  \caption{\label{fig:ex3}}
\end{wrapfigure}
We again apply our method \condSafe to find a conditional invariant
for this loop which, together with the loop condition, implies the
assertion in the loop body. As it can only synthesize conjunctions of linear
inequalities, it produces the conditional invariant $\QQ_3 = x > y$ for
the loop. However, assume that the precondition $\QQ_3$ could not be
proven to always hold in the context of our example. In that case, we
use the obtained information to narrow the program and look for
another precondition.

\begin{wrapfigure}[8]{r}{3.3cm}
  \vspace{-4ex}
  \begin{tabular}{l}
    \textbf{if} $\neg(\var{x} \!>\! \var{y})$ \textbf{then} \\
    \quad\textbf{while} $\neg(\var{x} \!>\! \var{y})$ \!\textbf{do} \\
    \quad\quad\textbf{assert}($\var{x} \neq \var{y}$);\\
    \quad\quad$\var{x} := \var{x} + 1$;\\
    \quad\quad$\var{y} := \var{y} + 1$;\\
    \quad\textbf{done}\\
    \textbf{fi}\\
  \end{tabular}
  \caption{\label{fig:ex3-strengthened}}
\end{wrapfigure}
Intuitively, our program narrowing reflects that states represented by
the conditional invariant found earlier are already proven to be safe.
Hence, we only need to consider states for which the negation of the
conditional invariant holds, i.e., we can add its negation as an
assumption to the program. In our example, this yields the modified
version of \rF{fig:ex3} displayed in \rF{fig:ex3-strengthened}.
Another call to \condSafe then yields the conditional invariant
$\QQ_3' = x < y$ for the loop. This means that we can ensure the
validity of the assertion if before the conditional statement we
satisfy that $\neg(x > y) \limplies x < y$, or equivalently, $x \neq
y$. In general, this narrowing allows us to find (some) disjunctive
invariants.


\section{Preliminaries}
\label{sect:Prelims}
\setcounter{paragraph}{0}

\subsubsection{\sat, \maxsat, and \maxsmt}
Let $\cP$ be a fixed set of \emph{propositional variables}. For $p \in
\cP$, $p$ and $\lnot p$ are \emph{literals}.
A \emph{clause} is a disjunction of literals $l_1 \lor\cdots\lor l_n$.
A (CNF) \emph{propositional formula} is a conjunction of clauses $C_1
\land\cdots\land C_m$.  The problem of \emph{propositional
  satisfiability} (\emph{\sat}) is to determine whether a
propositional formula is \emph{satisfiable}.  An extension of \sat is
\emph{satisfiability modulo theories (SMT)}~\cite{HandbookOfSAT2009},
where satisfiability of a formula with literals from a given
background theory is checked.  We will use the theory of
\emph{quantifier-free integer (non-)linear arithmetic}, where literals
are inequalities of linear (resp. polynomial) arithmetic expressions.

Another extension of \sat is \emph{\MAXSAT}~\cite{HandbookOfSAT2009}, which
generalizes \sat to
finding an assignment such that the number of satisfied clauses in a given formula $F$ is
maximized.
Finally, \emph{\MAXSMT} combines \maxsat and
\smt. 
A \emph{(weighted partial) \maxsmt} problem is a formula of the form 
 $H_1 \land \ldots \land H_n \land [S_1,\weight_1] \land \ldots \land [S_m,\weight_m]$,
where the hard clauses $H_i$ and the
soft clauses $S_j$ (with weight $\weight_j$) are disjunctions of 
literals over a background theory,
and the aim is to find a model
of the hard clauses that maximizes the sum of the weights of the satisfied soft clauses.

\subsubsection{Programs and States}
We make heavy use of the program structure and hence represent programs as
graphs. For this, we fix a set of (integer) program \emph{variables} $\VV = \{
v_1,\ldots,v_n \}$ and denote 
by $\QF(\VV)$ the formulas consisting of
conjunctions of linear inequalities over the variables $\VV$. Let
$\LL$ be the set of program \emph{locations}, which contains a
\emph{canonical start location} $\ell_0$.  Program \emph{transitions} $\TT$ are
tuples $(\ell, \tau, \ell')$, where $\ell$ and $\ell' \in \LL$
represent the pre- and post-location respectively, and $\tau \in
\QF(\VV \cup \VV')$ describes its transition relation. Here $\VV' = \{
v_1', \ldots, v_n'\}$ are the \emph{post-variables}, i.e., the
values of the variables after the transition.\footnote{For $\varphi
  \in \QF(\VV)$, $\varphi' \!\in\! \QF(\VV')$ is the version
  of $\varphi$ using primed variables.} A transition is \emph{initial}
if its source location is $\ell_0$.
A \emph{program} is a set of
transitions. We view a program $\PP = (\LL, \TT)$ as a directed graph
(the \emph{control-flow graph}, CFG), in which edges are the transitions 
$\TT$ and nodes are the locations $\LL$.%
\footnote{Since we label transitions only with conjunctions of linear
  inequalities, disjunctive conditions are represented using several
  transitions with the same pre- and post-location. Thus, $\PP$ is actually a
  multigraph.}

A \emph{state} $s = (\ell, \vect{v})$ consists of a location $\ell \in
\LL$ and a \emph{valuation} $\vect{v}: \VV \to \mathbb{Z}$.  A state
$(\ell, \vect{v})$ is \emph{initial} if $\ell = \ell_0$. We denote
an \emph{evaluation step} with transition $t = (\ell, \tau, \ell')$ by
$(\ell, \vect{v}) \to_t (\ell', \vect{v'})$, where the valuations
$\vect{v}$, $\vect{v'}$ satisfy the formula $\tau$ of $t$. We use
$\to_{\PP}$ if we do not care about the executed transition, and
$\to^*_{\PP}$ to denote the transitive-reflexive closure of
$\to_{\PP}$. We say that a state $s$ is \emph{reachable} if there
exists an initial state $s_0$ such that $s_0 \to^*_{\PP} s$.

\subsubsection{Safety and Invariants}
An \emph{assertion} $(t, \varphi)$ is a pair of a transition $t \in \TT$ and a
formula $\varphi \in \QF(\VV)$.
A program $\PP$ is \emph{safe} for the assertion $(t, \varphi)$ if for every
evaluation $(\ell_0, \vect{v}_0) \to^*_{\PP} \circ 
\to_t (\ell, \vect{v})$, we have that $\vect{v} \models \varphi$ holds.%
\footnote{Here, $\to^*_{\PP} \circ \to_t$ denotes arbitrary program evaluations that end
with an evaluation step using $t$.}
Note that proving that a formula $\varphi$ always holds at a location
  $\ell$ can be handled in this setting by adding an extra
location $\ell^*$ and an extra transition
$t^* = (\ell, \mathrm{true}, \ell^*)$ and checking safety for $(t^*, \varphi)$.

  We call a map $\cI: \LL \to \QF(\VV)$ a \emph{program invariant} (or often just
\emph{invariant}) if for all reachable states $(\ell, \vect{v})$, we have $\vect{v}
\models \cI(\ell)$ holds.
An important class of program invariants are inductive invariants.
An invariant $\cI$ is \emph{inductive} if the following conditions hold:
\begin{center}
 \begin{tabular}{lrl}
     \textbf{Initiation: } 
   & $\true$
   & $\models \cI(\ell_0)$\\
     \textbf{Consecution: }
   & For $(\ell, \tau, \ell') \in \cP$:\quad
     $\cI(\ell) \land \tau$
   & $\models \cI(\ell')'$
 \end{tabular}
\end{center}

\subsubsection{Constraint Solving for Verification}
Inductive invariants can be generated using a \emph{constraint-based}
approach~\cite{Colonetal2003CAV,Bradley05a}. 
The idea is to consider templates for candidate invariant properties, such as
(conjunctions of) linear inequalities.
These templates contain both the program variables $\VV$ as well as template
variables $\VV_T$, whose values have to be determined to ensure the
required properties.
To this end, the conditions on inductive invariants are expressed by means of
\emph{constraints} of the form $\exists \VV_T . \forall \VV . \ldots$.
Any solution to these constraints then yields an invariant.
In the case of linear arithmetic, Farkas' Lemma~\cite{Schrijver} is often used
to handle the quantifier alternation in the generated constraints.
Intuitively, it allows one to transform $\exists\forall$ problems encountered in
invariant synthesis into $\exists$ problems.
In the general case, an \smt problem over non-linear arithmetic is
obtained, for which effective \smt solvers exist
\cite{Borrallerasetal2011JAR,JovanovicMoura2012IJCAR}.
By assigning weights to the different conditions, invariant generation can be
cast as an optimization problem in the \maxsmt framework~\cite{Larraz13,Larraz14}.


\section{Proving Safety}
\label{sect:safety}

Most automated techniques for proving program safety iteratively construct
\emph{inductive program invariants} as over-approximations of the reachable
state space.
Starting from the known set of initial states, a process to discover more
reachable states and refine the approximation is iterated, until it finally
reaches a fixed point (i.e., the invariant is inductive) and is strong enough
to imply program safety.
However, this requires taking the whole program into account, which is sometimes
infeasible or undesirable in practice.

In contrast to this, our method starts with the known unsafe states, and
iteratively constructs an under-approximation of the set of safe states, with
the goal of showing that all initial states are contained in that set.
For this, we introduce the notion of \emph{conditional safety}.
Intuitively, when proving that a program is \emph{$(\tilde{t},
  \tilde\varphi)$-conditionally safe for the assertion} $(t, \varphi)$
we consider evaluations starting after a $\to_{\tilde{t}} (\tilde\ell,
\tilde{\vect{v}})$ step, where $\tilde{\vect{v}}$ satisfies $\tilde\varphi$,
instead of evaluations starting at an initial state.
In particular, a program that is $(t_0, \true)$-conditionally safe for $(t,
\varphi)$ for all initial transitions $t_0$ is (unconditionally) safe for $(t,
\varphi)$.

\begin{definition}[Conditional safety]
  Let $\PP$ be a program, $t, \tilde{t}$ transitions and $\varphi,
  \tilde\varphi \in \QF(\VV)$.
  The program $\PP$ is \emph{$(\tilde{t}, \tilde\varphi)$-conditionally safe} for
  the assertion $(t, \varphi)$ if for any evaluation that contains
  $\to_{\tilde{t}} (\tilde\ell, \tilde{\vect{v}}) \to_{\PP}^* (\bar\ell,
  \bar{\vect{v}}) \to_t (\ell, \vect{v})$, we have
  $\tilde{\vect{v}} \models \tilde\varphi$ implies that $\vect{v} \models \varphi$.
  In that case we say that the assertion $(\tilde{t}, \tilde\varphi)$ is a
  \emph{precondition} for the \emph{postcondition} $(t, \varphi)$.
\end{definition}

Conditional safety is ``transitive'' in the sense that if a set of
transitions $\cE = \{ \widetilde{t_1}, \ldots, \widetilde{t_m}\}$
dominates $t$,\footnote{We say a set of transitions $\cE$ dominates
  transition $t$ if every path in the CFG from $\ell_0$ that contains
  $t$ must also contain some $\tilde{t}\in \cE$.} and for all $i=1,
\ldots, m$ we have $\PP$ is
$(\widetilde{t_i},\widetilde\varphi_i)$-conditionally safe for $(t,
\varphi)$ and $\PP$ is safe for
$(\widetilde{t_i},\widetilde\varphi_i)$, then $\PP$ is also safe for
$(t, \varphi)$. In what follows we exploit this observation to prove
program safety by means of conditional safety.

A \emph{program component} $\cC$ of a program $\PP$ is an SCC of the
control-flow graph, and its \emph{entry transitions} (or
\emph{entries}) $\cE_\cC$ are those transitions $t = (\ell, \tau,
\ell')$ such that $t \not\in \cC$ but $\ell'$ appears in $\cC$.
By considering each component as a single node, we can obtain from $\PP$ a DAG
of SCCs, whose edges are the entry transitions.
Our technique analyzes components independently, and communicates the results of
these analyses to \ignore{the analysis of }other components along entry transitions.

Given a component $\cC$ and an assertion $(t, \varphi)$ such that $t\not\in \cC$
but the source node of $t$ appears in $\cC$, we call $t$ an \emph{exit
  transition} of $\cC$.
For such exit transitions, we compute a \emph{sufficient} condition
$\psi_{\tilde{t}}$ for each entry transition $\tilde{t} \in \cE_{\cC}$ such that
$\cC \cup \{ t \}$ is $(\tilde{t}, \psi_{\tilde{t}})$-conditionally safe for $(t,
\varphi)$. Then we continue reasoning backwards following the DAG and
try to prove that $\PP$ is safe for each $(\tilde{t}, \psi_{\tilde{t}})$.
If we succeed, following the argument above we will have proved $\PP$
safe for $(t, \varphi)$.

In the following, we first discuss how to prove conditional safety of single
program components in \rSC{sect:Conditional-Safety}, and then present the
algorithm that combines these local analyses to construct a global
safety proof in \rSC{sect:Proving-Safety}.

\subsection{Synthesizing Local Conditions}
\label{sect:Conditional-Safety}

Here we restrict ourselves to a program component $\cC$ and its entry
transitions $\entries_{\cC}$, and assume we are given an assertion
$(t_{\mathrm{exit}}, \varphi)$, where $t_{\mathrm{exit}} =
(\tilde\ell_{\mathrm{exit}}, \tau_{\mathrm{exit}},
\ell_{\mathrm{exit}})$ is an exit transition of $\cC$ (i.e.,
$t_{\mathrm{exit}} \not\in \cC$ and $\tilde\ell_{\mathrm{exit}}$
appears in $\cC$). We show how a
precondition $(t, \psi)$ for $(t_{\mathrm{exit}}, \varphi)$ can
be obtained for each $t \in \cE_\cC$. Here we only
consider the case of $\varphi$ being a single clause (i.e., a
disjunction of literals); if $\varphi$ is in CNF, each conjunct is handled
separately.
The preconditions on the
entry transitions will be determined by a \emph{conditional
inductive invariant}, which like a standard invariant is inductive, but not
necessarily initiated in all program runs.
Indeed, this initiation condition is what we will extract as precondition and
propagate backwards to preceding program components in the DAG.

\begin{definition}[Conditional Inductive Invariant]
 We say a map $\QQ: \LL \to \QF(\VV)$ is a \emph{conditional (inductive)
   invariant} for a program (component) $\PP$ if for all $(\ell, {\vect{v}})
 \to_{\PP} (\ell', \vect{v}')$, we have ${\vect{v}} \models \QQ(\ell)$ implies
 $\vect{v}' \models \QQ(\ell')$.
\end{definition}

Conditional invariants are convenient tools to express conditions for safety
proving, allowing reasoning in the style of ``if the condition for $\QQ$ holds,
then the assertion $(t, \varphi)$ holds''.

\medskip
\begin{figure}[t]
\begin{minipage}[c]{2.5cm}
    \vspace{-2.5cm}

    {\textcolor{white}{xxxxxx}}\begin{tabular}{l}
      \ldots\\
      \textbf{while} $\var{i} > 0$ \textbf{do} \\
      \quad$\var{x} := \var{x} + 5$;\\
      \quad$\var{i} := \var{i} - 1$;\\
      \textbf{done}\\
      \textbf{assert}($\var{x} \geq 50$)
    \end{tabular}
  \end{minipage}
  \begin{minipage}[t]{6cm}
    \small
{\textcolor{white}{xxxxxxxx}}
    \begin{tikzpicture}[cfg]
        \node[state] (l1) at (0,0) {$\ell_1$};
        \node[state] (l2) at ($(l1) + (0,-1.5)$) {$\ell_2$};

        \path[->]
          ($(l1.north) + (0,0.3)$) edge (l1)
          (l1) edge node[left]
            {$
               \begin{array}{r@{\,}l}
                       & \var{i} \leq 0\\
                 \land & \var{x}' = \var{x}\\
                 \land & \var{i}' = \var{i}
               \end{array}
             $}
           (l2)
          (l1) edge[loop right] node[right,xshift=-3mm,yshift=-.7cm]
            {$
               \begin{array}{r@{\,}l}
                       & \var{i} > 0\\
                 \land & \var{x}' = \var{x} + 5\\
                 \land & \var{i}' = \var{i} - 1
               \end{array}
             $}
           (l1)
          ;
      \end{tikzpicture}
    \end{minipage}
    \caption{Source code of program snippet and its CFG.}
    \label{fig:code_and_cfg}
  \end{figure}
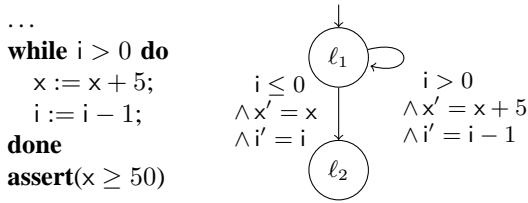
    
    \begin{example}
Consider the program snippet in \rF{fig:code_and_cfg}.
A conditional inductive invariant supporting safety of this program part is
$\QQ_5(\ell_1) \equiv \var{x} + 5 \cdot \var{i} \geq 50$, $\QQ_5(\ell_2) \equiv
\var{x} \geq 50$.
In fact, any conditional invariant $\QQ_m(\ell_1) \equiv \var{x} + m \cdot \var{i} \geq 50$
with $0 \leq m \leq 5$ would be a conditional inductive invariant that, together
with the negation of the loop condition $\var{i} \leq 0$, implies
$\var{x} \geq 50$.
\medskip
\end{example}


We use a \maxsmt-based constraint-solving approach to generate
conditional inductive invariants. Unlike in \cite{Larraz14}, to use information about the
initialization of variables before a program component, we take into account the
entry transitions $\entries_{\cC}$. The precondition for each entry
transition is the conditional invariant that has been synthesized
at its target location.

\begin{figure*}[t]
  \begin{center}
    \begin{tabular}{l@{\quad}l@{\quad}l@{\quad}l@{\quad}rl}
        \textbf{\footnotesize Initiation:}
        & For $t=(\ell, \tau, \ell') \in \entries_{\cC}, 1 \leq j \leq k$:
      & $\initiation_{t,j,k}$
      & $\define$
      & $\tau$
      & $\limplies I'_{\ell',j,k}$\\
        \textbf{\footnotesize Consecution:}
      & For $t=(\ell, \tau, \ell') \in \cC$:
      & $\consecution_{t,k}$
      & $\define$
      & $I_{\ell,k} \land \tau $
      & $\limplies I'_{\ell',k}$\\
        \textbf{\footnotesize Safety:}
      & For $t_{\mathrm{exit}} =
(\tilde\ell_{\mathrm{exit}}, \tau_{\mathrm{exit}},
\ell_{\mathrm{exit}})$:
      & $\safety_k$
      & $\define$
      & $I_{\tilde\ell_{\mathrm{exit}},k} \land \tau_{\mathrm{exit}} $
      & $\limplies \varphi'$\\
    \end{tabular}
  \end{center}
  \caption{\label{fig:condSafe-constraints}Constraints used in
    $\condSafe(\cC, \entries_{\cC}, (t_{\mathrm{exit}}, \varphi))$}
\end{figure*}

To find conditional invariants, we construct a constraint system.
For each location $\ell$ in $\cC$ we create a template
 $I_{\ell,k}(\vars) \equiv \wedge_{1 \leq j \leq k}\; I_{\ell,j,k}(\vars)$
which is a conjunction of $k$ linear inequations\footnote{In our overall
  algorithm, $k$ is initially 1 and increased in case of failures.}
 of the form $I_{\ell,j,k}(\vars) \equiv i_{\ell,j} + \sum_{v \in
  \vars} i_{\ell,j,v} \cdot v \leq 0$,
where the $i_{\ell,j}$, $i_{\ell,j,v}$ are fresh variables from the set of
template variables $\VV_T$.
We then transform the conditions for a conditional invariant proving safety for
the assertion $(t_{\mathrm{exit}}, \varphi)$ to the constraints in
\rF{fig:condSafe-constraints}.
Here, e.g., $I'_{ \ell',k}$ refers to the variant of $I_{\ell',k}$ using primed
versions of the program variables $\VV$, but unprimed template variables
$\VV_T$.

In the overall constraint system, we mark the \textbf{Consecution} and
\textbf{Safety} constraints as hard requirements.
Thus, any solution to these constraints is a conditional \emph{inductive}
invariant \emph{implying our assertion}.
However, as we mark the \textbf{Initiation} constraints as soft, the found
conditional invariants may depend on preconditions not implied by the direct context of the
considered component. On the other hand, the \maxsmt solver prefers solutions that require
fewer preconditions.
Overall, we create the following \maxsmt formula
$$
\formula_k \define\! \bigwedge_{t \in \cC} \consecution_{t,k}
  \;\land\!\!\!\!\!\!\!
\bigwedge_{t \in \entries_{\cC}, 1 \leq j \leq k} \!\!\!\!\!\!\!\!\!\big(\initiation_{t,j,k} \,\lor\, \lnot \, \VarInitiation{t,j,k} \big)
  \,\land
\,\safety_k\,
  \land \!\!\!\!\!\!\! \bigwedge_{t \in \entries_{\cC}, 1 \leq j \leq k} \!\!\!\!\!\!\![\VarInitiation{t,j,k}, \weightInitiation],
$$
where the $\VarInitiation{t,j,k}$ are propositional variables which are true if
the \textbf{Initiation} condition $\initiation_{t,j,k}$ is satisfied, and
$\weightInitiation$ is the corresponding weight.
\footnote{Farkas' Lemma is applied \emph{locally} to the subformulas 
$\consecution_{t,k}$, $\initiation_{t,j,k}$ and $\safety_k$, and weights are added on the resulting 
constraints over the template variables.}
We use $\formula_k$ in our procedure \condSafe in
\rA{alg:safety-precondition}.
\begin{algorithm}
  \begin{algorithmic}[1]
    \Require component $\cC$, entry transitions $\entries_{\cC}$,
             assertion $(t_{\mathrm{exit}}, \varphi)$ s.t. $t_{\mathrm{exit}}$ is an exit transition of $\cC$ and $\varphi$ is a clause

    \Ensure $\textsf{None} \mid \cQ$, where $\cQ$ maps locations in $\cC$ to conjunctions of inequations
    \State{$k \leftarrow 1$}
    \Repeat
      \State{construct formula $\formula_k$ from $\cC$, $\entries_{\cC}$ and $(t_{\mathrm{exit}}, \varphi)$}
      \State{$\sigma \leftarrow \textsf{\maxsmt-solver}(\formula_k)$}
      \If{$\sigma \textrm{ is a model}$}
        \State{$\cQ \leftarrow \{ \ell \mapsto \sigma(I_{\ell,k}) \mid \ell \text{ in } \cC \}$}
        \Return $\cQ$ 
      \EndIf
      \State{$k \leftarrow k + 1$}
    \Until {$k > \textsf{MAX\_CONJUNCTS}$ }
    \Return \textsf{None}
  \end{algorithmic}
  \caption{Proc. \condSafe computing conditional invariant}
\label{alg:safety-precondition}
\end{algorithm}

In \condSafe, we iteratively try ``larger'' templates of more conjuncts
of linear inequations until we either give up (in our implementation,
\textsf{MAX\_CONJUNCTS} is 3) or finally find a conditional invariant.
Note, however, that here we are only trying to prove safety for \emph{one}
clause at a time,
which reduces the number of required conjuncts as compared to dealing
with a whole CNF in a single step.
If the \maxsmt solver is able to find a model for $\formula_k$, then we
instantiate our invariant templates $I_{\ell,k}$ with the values found for the
template variables in the model $\sigma$, obtaining a conditional invariant $\cQ$.
When we obtain a result, for every entry transition $t = (\ell, \tau, \ell') \in \cE_\cC$ the conditional invariant $\cQ(\ell')$ is a
precondition that implies safety for the assertion
$(t_{\mathrm{exit}}, \varphi)$. The following theorem states the correctness of this procedure.

\begin{theorem}
  \label{thm:condSafe-sound}
  Let $\cC$ be a component, $\entries_{\cC}$ its entry transitions,
  and $(t_{\mathrm{exit}}, \varphi)$ an assertion with
  $t_{\mathrm{exit}}$ an exit transition of $\cC$ and $\varphi$ a
  clause. If the procedure call $\condSafe(\cC, \entries_{\cC},
  (t_{\mathrm{exit}}, \varphi))$ returns $\cQ \neq \mathsf{None}$,
  then $\cQ$ is a conditional inductive invariant for $\cC$ and $\PP$
  is $(t, \cQ({\ell'}))$-conditionally safe for $(t_{\mathrm{exit}},
  \varphi)$ for all $t = (\ell, \tau, \ell') \in \entries_{\cC}$.
\end{theorem}

\begin{proof}
 {All proofs can be found in Appendix \ref{sect:proofs}.}
\end{proof}

\subsection{Propagating Local Conditions}
\label{sect:Proving-Safety}

In this section, we explain how to use the local procedure \condSafe
to prove safety of a full program. To this end we now consider the
full DAG of program components. As outlined above, the idea is to
start from the assertion provided by the user, call the procedure
\condSafe to obtain preconditions for the entry transitions of the
corresponding component, and then use these preconditions as
assertions for preceding components, continuing recursively. If
eventually for each initial transition the transition relation implies
the corresponding preconditions, then safety has been proven. If we
fail to prove safety for certain assertions, we backtrack, trying
further possible preconditions and conditional invariants.

The key to the precision of our approach is our treatment of failed
proof attempts. When the procedure \condSafe finds a conditional
invariant $\QQ$ for $\cC$, but proving $(t, \QQ(\ell'))$ as a postcondition of
the preceding component fails for some $t = (\ell, \tau, \ell') \in
\cE_\cC$, we use $\QQ$ to \emph{narrow} our program representation
and filter out evaluations that are already known to be safe.

\begin{algorithm}[t]
\begin{algorithmic}[1]
 \Require Program $\PP$, a
          component $\cC$,
          entries $\entries_{\cC}$,
          assertion $(t_{\mathrm{exit}}, \varphi)$ s.t. 
            $t_{\mathrm{exit}}$ is an exit transition of $\cC$ and
            $\varphi$ a clause

 \Ensure $\textsf{Safe} \mid \textsf{Maybe}$
 \State \textbf{let} $(\ell_{\mathrm{exit}}, \tau_{\mathrm{exit}}, \ell'_{\mathrm{exit}}) = t_{\mathrm{exit}}$
 \If{$(\tau_{\mathrm{exit}} \limplies \varphi')$}
    \Return $\textsf{Safe}$
    \ElsIf
    {$\ell_{\mathrm{exit}} = \ell_0$}
    \Return \textsf{Maybe}
 \EndIf
 \State $\QQ \leftarrow \condSafe(\cC, \entries_{\cC}, (t_{\mathrm{exit}}, \varphi))$ 
 \If{$\QQ = \textsf{None}$}
   \Return \textsf{Maybe}
 \EndIf
 \ForAll
   {$t = (\ell, \tau, \ell') \in \entries_{\cC}$, $L \in \QQ(\ell')$}
   \State $\tilde{\cC} \leftarrow \textsf{component}(\ell, \PP)$
   \State $\entries_{\tilde{\cC}} \leftarrow \textsf{entries}(\tilde{\cC},\PP)$
   \State $\textsf{res}[t, L] \leftarrow \checkSafe(\PP, \tilde{\cC}, \entries_{\tilde{\cC}}, (t, L))$
 \EndFor
 \If{$\forall t = (\ell, \tau, \ell') \in \entries_{\cC}, L \in \QQ(\ell') \;.\; \textsf{res}[t, L] = \textsf{Safe}$}\\
   \Return \textsf{Safe} 
 \Else
   \State $\hat{\entries}_{\cC} \leftarrow
      \{ (\ell, \tau \land \lnot
                  (\displaystyle\!\!\!\!\bigwedge_{L \in \cQ({\ell'}) \atop
                    \textsf{res}[t,L] = \textsf{Maybe}}\!\!\!\! L')
              , \ell') \mid t = (\ell, \tau, \ell') \in \entries_{\cC}
      \}$ 
   \State $\hat{\cC} \leftarrow
      \{ (\ell, \tau \land \lnot \cQ({\ell'})' \land \lnot \cQ({\ell})
              , \ell') \mid (\ell, \tau, \ell') \in \cC
      \}$ \ignore{\COMMENT{Narrow component}}\\
   \Return $\checkSafe(\PP,\hat{\cC}, \hat{\entries}_{\cC}, (t_{\mathrm{exit}}, \varphi))$
 \EndIf
\end{algorithmic}
  \caption{Procedure \checkSafe for proving a program safe for an assertion}
\label{alg:safety-proving}
\end{algorithm}

As outlined above, in our proof process we treat each clause of the
conjunction $\QQ(\ell')$ separately, and pass each one as its own
assertion to preceding program components, allowing for a fine-grained
\emph{program-narrowing} technique. By construction of $\QQ$, evaluations
that satisfy all literals of $\QQ(\ell')$ after executing $t = (\ell, \tau,
\ell') \in \cE_\cC$ are safe. Thus, among the evaluations that
use $t$, we only need to consider those where at least one
literal in $\QQ(\ell')$ does not hold. Hence, we \emph{narrow} each entry
transition by conjoining it with the negation of the conjunction of
all literals for which we could not prove safety (see line 13 in
\rA{alg:safety-proving}). Note that if there is more than one literal
in this conjunction, then the negation is a disjunction,
which in our program model implies splitting transitions.

We can narrow program components similarly. For a transition $t = (\ell,
\tau, \ell') \in \cC$, we know that if either $\QQ(\ell)$ or
$\QQ(\ell')'$ holds in an evaluation passing through $t$, the program is safe.
Thus, we narrow the program by replacing $\tau$ by $\tau \land \neg\QQ(\ell)
\land \neg\QQ(\ell')'$ (see line 14 in \rA{alg:safety-proving}).

This narrowing allows us to generate disjunctive conditional invariants, where
each result of \condSafe is one disjunct.
Note that not \emph{all} disjunctive invariants can be discovered like this, as
each intermediate result needs to be inductive using the disjuncts found so
far.
However, this is the pattern observed in \emph{phase-change}
algorithms~\cite{Sharma11}.

Our overall safety proving procedure \checkSafe is shown in
\rA{alg:safety-proving}.
The helper procedures \textsf{component} and \textsf{entries} are used to
find the program component for a given location and the entry transitions for a
component.
The result of \checkSafe is either \textsf{Maybe} when the proof failed, or
\textsf{Safe} if it succeeded.
In the latter case, we have managed to create a chain of conditional
invariants that imply that $(t_{\mathrm{exit}}, \varphi)$ always holds.

\begin{theorem}
\label{thm:checkSafe-sound}
Let $\PP$ be a program, $\cC$ a component and $\entries_{\cC}$ its
entries. Given an assertion $(t_{\mathrm{exit}}, \varphi)$ such that
$t_{\mathrm{exit}}$ is an exit transition of $\cC$ and $\varphi$ is a
clause, if $\checkSafe(\PP, \cC, \cE_\cC, (t_{\mathrm{exit}},
\varphi)) = \textsf{Safe}$, then $\PP$ is safe for
$(t_{\mathrm{exit}}, \varphi)$.
\end{theorem}

\begin{example}
  We demonstrate \checkSafe on the program displayed on
  \rF{fig:ex3ext}, called $\PP$ in the following, which is
  an extended version of the example from \rF{fig:ex3}.

  We want to prove the assertion $(t_5, \var{x} \!\neq\! \var{y})$. Hence we make a 
  first call $\checkSafe(\PP, \{ t_4 \}, \{ t_3 \}, (t_5, \var{x} \!\neq\! \var{y}))$:
  the non-trivial SCC containing $\ell_2$ is $\{ t_4 \}$ and its entry
  transitions are $\{ t_3 \}$. Hence, we call
    $\condSafe(\{ t_4 \}, \{ t_3 \}, (t_5, \var{x} \neq \var{y}))$
  and the resulting conditional invariant for $\ell_2$ is either
  $\var{x} < \var{y}$ or $\var{y} < \var{x}$. Let us assume it is
  $\var{y} < \var{x}$.
  In the next step, we propagate this to the predecessor SCC $\{ t_2 \}$, and
  call $\checkSafe(\PP, \{ t_2 \}, \{ t_1 \}, (t_3, \var{y} < \var{x}))$.

  In turn,
  this leads to calling
    $\condSafe(\{ t_2 \}, \{ t_1 \}, (t_3, \var{y} < \var{x}))$
  to our synthesis subprocedure. No conditional invariant supporting this
  assertion can be found, and hence \textsf{None} is returned by
  \condSafe, and consequently \textsf{Maybe} is returned by
  \checkSafe.
  Hence, we return to the original SCC $\{ t_4 \}$ and its entry $\{ t_3 \}$, and then by narrowing
  we obtain two new transitions:
  \begin{align*}
    t_4' & = (\ell_2, \var{x}' = \var{x} + 1 \land \var{y}' = \var{y} + 1 \land \neg(\var{y} < \var{x}), \ell_2),\\
    t_3' & = (\ell_1, \var{x} < 0 \land \var{x}' = \var{x} \land \var{y}' = \var{y} \land \neg(\var{y} < \var{x}), \ell_2).
  \end{align*}
  Using these, we call $\checkSafe(\PP, \{ t_4' \}, \{ t_3' \}, (t_5, \var{x} \!\neq\! \var{y}))$.
  The next call to \condSafe then yields the conditional invariant
  $\var{x} < \var{y}$ at $\ell_2$, which is in turn propagated backwards with the call
  $\checkSafe(\PP, \{ t_2 \}, \{ t_1 \}, (t_3', \var{x} < \var{y}))$.
  This then yields a conditional invariant $\var{x} < \var{y}$ at $\ell_1$, which is
  finally propagated back in the call $\checkSafe(\PP, \{ \}, \{ \}, (t_1, \var{x} < \var{y}))$,
  which directly returns $\textsf{Safe}$.
\end{example}

\subsection{Improving Performance}
\setcounter{paragraph}{0}

\begin{wrapfigure}[13]{r}{3.5cm}
  \vspace*{-10ex}
  \hspace*{-3ex}
  \begin{tikzpicture}[cfg]
    \node[state] (l0) at (0,0) {$\ell_0$};
    \node[state] (l1) at ($(l0) + (0,-2)$) {$\ell_1$};
    \node[state] (l2) at ($(l1) + (0,-2)$) {$\ell_2$};
    \node[state] (l3) at ($(l2) + (0,-1.5)$) {$\ell_3$};

    \path[->]
      ($(l0.west) + (-1, 0)$) edge (l0)
      (l0) edge
        node[left, yshift=0.5ex]
          {$
            \begin{small}
            \begin{array}{r@{\,}r@{\,}l}
             t_1:&       & \var{x} < \var{y}\\
                 & \land & \var{x}' = \var{x}\\
                 & \land & \var{y}' = \var{y}
            \end{array}
            \end{small}
           $}
        (l1)
      (l1) edge[loop left]
        node[left]
          {$
            \begin{small}
            \begin{array}{r@{\,}r@{\,}l}
             t_2:\!\!\!\!\!\!&       & \var{x} \geq 0\\
                 & \land & \var{x}' = \var{x} - 1\\
                 & \land & \var{y}' = \var{y}
            \end{array}
            \end{small}
           $}
        (l1)
      (l1) edge
        node[left, xshift=-1.5ex, yshift=-1ex]
          {$
            \begin{small}
            \begin{array}{r@{\,}r@{\,}l}
             t_3:&       & \var{x} < 0\\
                 & \land & \var{x}' = \var{x}\\
                 & \land & \var{y}' = \var{y}
            \end{array}
            \end{small}
           $}
        (l2)
      (l2) edge[loop left]
        node[left, yshift=-1.5ex]
          {$
            \begin{small}
            \begin{array}{r@{\,}r@{\,}l}
             t_4:\!\!\!\!\!\!&       & \var{x}' = \var{x} + 1\\
                 & \land & \var{y}' = \var{y} + 1
            \end{array}
            \end{small}
           $}
        (l2)
      (l2) edge
        node[left, yshift=-1ex]
          {\begin{small}
              $t_5: \true$
            \end{small}
          }
        (l3)
    ;
  \end{tikzpicture}
  \vspace*{-4ex}
  \caption{\label{fig:ex3ext}}
\end{wrapfigure}
The basic method \checkSafe can be extended in several ways to improve
performance.
We now present a number of techniques that are useful to reduce
the runtime of the algorithm and distribute the required work.
Note that none of these techniques influences the precision of
the overall framework.

\paragraph{Using conditional invariants to disable transitions}

When proving an assertion, it is often necessary to find invariants
that show the unfeasibility of some transition, which allows disabling
it. In our framework, the required invariants can be conditional as
well. Therefore, \checkSafe must be called recursively to prove that
the conditional invariant is indeed invariant. In our implementation,
we generate constraints such that every solution provides conditional
invariants either implying the postcondition or disabling some
transition. By imposing different weights, we make the \maxsmt solver
prefer solutions that imply the postcondition.

\paragraph{Handling unsuccessful proof attempts}
One important aspect is that the presented algorithm does not \emph{learn} facts
about the reachable state space, and so duplicates work when assertions appear
several times.
To alleviate this for \emph{unsuccessful} recursive invocations of \checkSafe,
we introduce a simple memoization technique to avoid repeating such calls.
So when $\checkSafe(\PP, \cC, \cE_{\cC}, (t, \varphi)) = \textsf{Maybe}$, we
store this result, and use it for all later calls of $\checkSafe(\PP,
\cC, \cE_{\cC}, (t, \varphi))$.
This strategy is valid as the return value $\textsf{Maybe}$ indicates that our
method cannot prove the assertion $(t, \varphi)$ at all, meaning that later
proof attempts will fail as well. In our implementation, this
memoization of unsuccessful attempts is local to the initial call to
$\checkSafe$. The rationale is that, when proving unrelated properties, it is likely
that few calls are shared and that the book-keeping does not pay off.

\paragraph{Handling successful proof attempts}
When a recursive call yields a \emph{successful} result, we can \emph{strengthen}
the program with the proven invariant. Remember that $\checkSafe(\PP, \cC,
\cE_{\cC}, (t, \varphi)) = \textsf{Safe}$ means that whenever the transition $t$
is used in \emph{any} evaluation, $\varphi$ holds in the succeeding
state. Thus, we can add this knowledge explicitly and
change the transition in the original program. In practice,
this strengthening is applied only if the first call to \checkSafe was
successful, i.e, no narrowing was applied. The reason is that, if the transition relation of $t$ was
obtained through repeated narrowing,
in general one needs to split transitions, and
it is not correct to just add
$\varphi'$ to $t$.
Namely, assume that $t_o = (\ell, \tau_o, \ell')$ is the original (unnarrowed)
version of a transition $t = (\ell, \tau, \ell') \in \cE_\cC$.
As $t$ is an entry transition of $\cC$, we have $\tau = \tau_o \land \lnot
\psi_1' \land \ldots \land \lnot \psi_m'$ by construction, where $\psi_i$ is the
additional constraint we added in the $i$-th narrowing of component entries.
Thus what we proved is that $\psi_1' \lor \ldots \lor \psi_m' \lor \varphi'$ always
holds after using transition $t_o$.
So we should replace $t_o$ in the program with a
transition labeled with $\tau_o \land (\psi_1' \lor \ldots \lor \psi_m' \lor
\varphi')$. As we cannot handle disjunctions natively, this implies
replacing $t_o$ by $m+1$ new transitions.

Note that this program modification approach, unlike memoization, 
makes the gained information available to the \maxsmt solver when
searching for a conditional invariant.
A similar strategy can be used to strengthen the transitions in the considered
component $\cC$.

\paragraph{Parallelizing \& distributing the analysis}

Our analysis can easily be parallelized. We have implemented this at
two stages. First, at the level of the procedure \condSafe, we try at the same time
different numbers of template conjuncts (lines 3-6 in
\rA{alg:safety-precondition}), which requires calling several instances of the
solver simultaneously. Secondly, at a higher level, the recursive calls of
\checkSafe (line 9 in \rA{alg:safety-proving}) are parallelized.
Note that, since narrowing and the ``learning'' optimizations
described above are considered only locally, they can be handled as
asynchronous updates to the program kept in each worker, and do not
require synchronization operations.
Hence, distributing the analysis onto several worker processes, in the style of
\textsf{Bolt}~\cite{Albarghouthi12b}, would be possible as well.

Other directions for parallelization, which have not been implemented
yet, are to return different conditional invariants in parallel when the \maxsmt
problem in procedure \condSafe has several solutions.
Moreover, based on experimental observations that successful safety proofs have a
short successful path in the tree of proof attempts, we are also interested in
exploring a look-ahead strategy: after calling \condSafe in \checkSafe, we could
make recursive calls of \checkSafe on some processes while others are
already applying narrowing.

\paragraph{Iterative proving}
Finally, one could store the conditional invariants generated during a
successful proof, which are hence invariants, so that they can be
re-used in later runs. E.g., if a single component is modified, one
can reprocess it and compute a new precondition that ensures its
postcondition. If this precondition is implied by the previously
computed invariant, the program is safe and nothing else needs to be
done. Otherwise, one can proceed with the preceding components, and
produce respective new preconditions in a recursive way. Only when
proving safety with the previously computed invariants in this way
fails, the whole program needs to be reprocessed again. This technique
has not been implemented yet, as our prototype is still in a
preliminary state.



\section{Related Work}
Safety proving is an active area of research. In the recent past, techniques
based on variations of counterexample guided abstraction refinement have
dominated~\cite{Ball01,Henzinger03,Clarke05,McMillan06,Podelski07,Bradley11,Grebenshchikov12,Albarghouthi12,Cimatti12,Hoder12}.
These methods prove safety by repeatedly unfolding the program relation using a
symbolic representation of program states, starting in the initial states.
This process generates an over-approximation of the set of reachable states,
where the coarseness of the approximation is a consequence of the used symbolic
representation.
Whenever a state in the over-approximation violates the safety condition, either
a true counterexample was found and is reported, or the approximation is refined
(using techniques such as predicate
abstraction~\cite{DBLP:conf/popl/FlanaganQ02} or Craig
interpolation~\cite{DBLP:conf/sas/McMillan03}).
When further unwinding does not change the symbolic representation, all
reachable states have been found and the procedure terminates.
This can be understood as a ``top-down'' (``forward'') approach (starting from
the initial states), whereas our method is ``bottom-up'' (``backwards''),
i.e. starting from the assertions.

Techniques based on Abstract Interpretation~\cite{Cousot77} have had
substantial success in the industrial setting.
There, an abstract interpreter is instantiated by an abstract domain whose
elements are used to over-approximate sets of program states.
The interpreter then evaluates the program on the chosen abstract domain,
discovering reachable states.
A widening operator, combining two given over-approximations to a more general
one representing both, is employed to guarantee termination of the analysis when
handling loops.

``Bottom-up'' safety proving with preconditions found by abduction has been
investigated in \cite{Dillig13}.
This work is closest to ours in its overall approach, but uses fundamentally
different techniques to find preconditions. 
Instead of applying \maxsmt, the approach uses an abduction engine based on
maximal universal subsets and quantifier elimination in Presburger arithmetic.
Moreover, it does not have an equivalent to our narrowing to exploit failed proof
attempts.\ignore{, though a syntactic version of it~\cite{Sharma11} could be combined
with the method.}
In a similar vein, \cite{Pasareanu04} uses straight-line weakest precondition
computation and backwards-reasoning to infer loop invariants supporting validity
of an assertion.
To enforce a generalization towards inductive invariants, a heuristic
syntax-based method is used.

Automatically constructing program proofs from independently obtained subproofs
has been an active area of research in the recent past.
Splitting proofs along syntactic boundaries (e.g., handling procedures
separately) has been explored in
\cite{DBLP:conf/popl/YorshYC08,Godefroid10,Calcagno11,Albarghouthi12b}.
For each such unit, a summary of its behavior is computed, i.e., an expression
that connects certain (classes of) inputs to outputs.
Depending on the employed analyzers, these summaries encode \ignore{inputs that lead to
errors~\cite{Godefroid07},} under- and over-approximations of reachable
states~\cite{Godefroid10} or changes to the heap using
separation logic's frame rule~\cite{Calcagno11}.
Finally, \cite{Albarghouthi12b} discusses how such compositional analyses can
leverage cloud computing environments to parallelize and scale up program
proofs.

\section{Implementation and Evaluation}

We have implemented the algorithms from \rSC{sect:Conditional-Safety} and
\rSC{sect:Proving-Safety} in our early prototype \tool{VeryMax}, using the
\maxsmt solver for non-linear arithmetic~\cite{DBLP:conf/sat/LarrazORR14} in
the \tool{Barcelogic}~\cite{DBLP:conf/cav/BofillNORR08} system.
We evaluated a sequential (\tool{VeryMax-Seq}) and a parallel
(\tool{VeryMax-Par}) variant on two benchmark sets.

The first set (which we will call \tool{HOLA-BENCHS}) are the 46
programs from the evaluation of safety provers in \cite{Dillig13}
(which were collected from a variety of sources, among others,
\cite{Gupta09,DBLP:conf/pldi/GulwaniSV08,DBLP:conf/pldi/BeyerHMR07,DBLP:conf/tacas/GulavaniCNR08,DBLP:journals/fac/BradleyM08,Mine06,DBLP:conf/tacas/JhalaM06,Sharma11,DBLP:conf/cav/SharmaNA12,DBLP:conf/sigsoft/GulavaniHKNR06,DBLP:conf/tacas/GulavaniR06,DBLP:conf/pldi/DilligDA12},
the NECLA Static Analysis Benchmarks, etc.). The programs are
relatively small (they have between 17 and 71 lines of code, and
between 1 and 4 nested or consecutive loops), but expose a number of
``hard'' problems for analyzers. All of them are safe.

On this first benchmark set we compare with three systems. The first two
were leading tools in the Software Verification Competition
2015~\cite{DBLP:conf/tacas/000115}:
\tool{CPAchecker}\footnote{We ran \textsf{CPAchecker} with two different configurations, \tool{predicateAnalysis} and \tool{sv-comp15}.}~\cite{Beyer11}, which was the overall winner and
in particular won the gold medal in the \emph{``Control Flow and
  Integer Variables''} category, and \tool{SeaHorn}~\cite{Kahsai15},
which got the silver medal, and also won the \emph{``Simple''} category. We
also compare with \tool{HOLA}~\cite{Dillig13}, an abduction-based
backwards reasoning tool. Unfortunately, we were not able to obtain an
executable for \tool{HOLA}. For this reason we have taken the
experimental data for this tool directly from \cite{Dillig13}, where
it is reported that the experiments were performed on an Intel i5 2.6
GHz CPU with 8 Gb of memory. For the sake of a fair comparison, we
have run the other tools on a 4-core machine with the same
specification, using the same timeout of 200 seconds.
\rTab{tab:comparison} summarizes the results, reporting the number of
successful proofs, failed proofs, and timeouts (TO), together with the
respective total runtimes. Both versions of \tool{VeryMax} are competitive, and
our parallel version was two times faster than our sequential one on four
cores. As a reference, on these
examples \tool{VeryMax-Seq} needed 2.8 overall calls (recursive or
after narrowings) on average, with a maximum of 16. The number
of narrowings was approximately 1, with a maximum of 13. Our memoization
technique making use of already failed proof attempts was employed in
about one third of the cases.

\begin{table*}[!t]
  \center
  \begin{tabular}{|l|rr|rr|r|r|}
    \hline
    Tool              & Safe & $\Sigma$ s & Fail & $\Sigma$ s & TO & Total s\\\hline
    \tool{CPAchecker sv-comp15} & 33   & 2424.41     & 3      & 61.28      & 10 &  4489.73 \\\hline
    \tool{CPAchecker predicateAnalysis} & 25   & 503.05     & 11      & 19.72      & 10 &  2271.12 \\\hline
    \tool{SeaHorn}    & 32   & 7.95       & 13      & 3.477      & 1  &   211.56 \\\hline    
    \tool{HOLA}       & 43   & 23.53      & 0       & 0          & 3  &   623.53 \\\hline
    \tool{VeryMax-Seq}& 44   & 293.14     & 2       & 50.69      & 0  &   343.83 \\\hline
    \tool{VeryMax-Par}& 45   & 138.40      & 1       & 12.81       & 0  &   151.21 \\\hline
  \end{tabular}\smallskip
  \caption{Experimental results on \tool{HOLA-BENCHS} benchmark set.\qquad\qquad\qquad\qquad\qquad\qquad\qquad\qquad\qquad\qquad\qquad\qquad\qquad\qquad\qquad\qquad\qquad\qquad\qquad\qquad\qquad\qquad}
  \label{tab:comparison}
  \vspace{-2em}
\end{table*}

In our second benchmark set (which we will refer to as
\tool{NR-BENCHS}) we have used integer abstractions of 217 numerical
algorithms from \cite{Press02}. For each procedure and for each array
access in it, we have created two safety problems with one assertion
each, expressing that the index is within bounds. In some few cases
the soundness of array accesses in the original program depends on
properties of floating-point variables, which are abstracted away. So
in the corresponding abstraction some assertions may not hold.
Altogether, the resulting benchmark suite
consists of 6452 problems, of up to 284 lines of C code. Due to the size of this set, and to give more room to exploit parallelism (both
tools with which we compare on these benchmarks, \tool{CPAchecker} and
\tool{SeaHorn}, make use of several cores), we performed the experiments with a
more powerful machine, namely, an 8-core Intel i7 3.4 GHz CPU with $16$ GB of
memory. The time limit is 300 seconds.

\begin{table*}[!t]
  \center
  \begin{tabular}{|l|rr|rr|rr|r|r|}
    \hline
    Tool                                & Safe & $\Sigma$ s & Unsafe & $\Sigma$ s & Fail    & $\Sigma$ s & TO & Total s  \\\hline
    \tool{CPAchecker sv-comp15}         & 5570 & 614803.98  &  251   &   6188.30  & 326     & 28749.78   &305 & 735336.82 \\\hline
    \tool{CPAchecker predicateAnalysis} & 5928 & 23417.15   &  170   &   495.13   & 234     & 9105.69    &120 & 64652.29 \\\hline
    \tool{SeaHorn}                      & 6077 & 4276.21    &  233   &   135.25   & 80      & 529.09     & 62 & 24167.11 \\\hline    
    \tool{VeryMax-Seq}                  & 6105 & 5940.88    &  0     &   0        & 326     & 26739.30   & 21 & 38980.80 \\\hline
    \tool{VeryMax-Par}                  & 6106 & 4789.73    &  0     &   0        & 346     & 18878.42   &  0 & 23668.15 \\\hline
  \end{tabular}\smallskip
  \caption{Experimental results on \tool{NR-BENCHS} benchmark set.}
  \label{tab:comparison2}
  \vspace{-2em}
\end{table*}

The results can be seen in \rTab{tab:comparison2}.
On these instances, \tool{VeryMax} is able to prove more assertions than any of
the other tools, while being about as fast as \tool{SeaHorn}, and significantly
faster than \tool{CPAchecker}.
Note that many examples are solved very quickly in the sequential solver
already, and thus do not profit from our parallelization.
\tool{VeryMax} is at an early stage of development, and is hence not yet fully
tuned. For example, a number of program slicing techniques have not been
implemented yet, which would be very useful for handling larger programs. Thus,
we expect that further development will improve the tool performance
significantly.
The benchmarks and our tool can be found at
\url{http://www.cs.upc.edu/~albert/VeryMax.html}.

\section{Conclusion}
We have presented a novel approach to compositional safety verification.
Our main contribution is a proof framework that refines intermediate results
produced by a \maxsmt-based precondition synthesis procedure.
In contrast to most earlier work, we proceed \emph{bottom-up} to compute
summaries of code that are guaranteed to be relevant for the proof.

We plan to further extend \tool{VeryMax} to cover more program features and
include standard optimizations (e.g., slicing and constraint propagation with
simple abstract domains). It currently handles procedure calls by inlining, and
does not support recursive functions yet. However, they can be handled
by introducing templates for function pre/postconditions.

In the future, we are interested in experimenting with alternative precondition
synthesis methods (e.g., abduction-based ones).
We also want to combine our method with a \maxsmt-based termination proving
method~\cite{Larraz13} and extend it to \emph{existential} properties such as
reachability and non-termination~\cite{Larraz14}.
We expect to combine all of these techniques in an \emph{alternating}
procedure~\cite{Godefroid10} that tries to prove properties at the same time
as their duals, and which uses partial proofs to narrow the state space that
remains to be considered.
Eventually, these methods could be combined to verify arbitrary temporal
properties.
In another direction, we want to consider more expressive theories to model
program features such as arrays or the heap.




\bibliographystyle{IEEEtran}

\begin{thebibliography}{10}
\providecommand{\url}[1]{#1}
\providecommand{\newblock}{\relax}
\providecommand{\bibinfo}[2]{#2}
\providecommand{\BIBentrySTDinterwordspacing}{\spaceskip=0pt\relax}
\providecommand{\BIBentryALTinterwordstretchfactor}{4}
\providecommand{\BIBentryALTinterwordspacing}{\spaceskip=\fontdimen2\font plus
\BIBentryALTinterwordstretchfactor\fontdimen3\font minus
  \fontdimen4\font\relax}
\providecommand{\BIBforeignlanguage}[2]{{%
\expandafter\ifx\csname l@#1\endcsname\relax
\typeout{** WARNING: IEEEtran.bst: No hyphenation pattern has been}%
\typeout{** loaded for the language `#1'. Using the pattern for}%
\typeout{** the default language instead.}%
\else
\language=\csname l@#1\endcsname
\fi
#2}}
\providecommand{\BIBdecl}{\relax}
\BIBdecl

\bibitem{Godefroid10}
P.~Godefroid, A.~V. Nori, S.~K. Rajamani, and S.~Tetali, ``Compositional
  may-must program analysis: unleashing the power of alternation,'' in
  \emph{POPL}, 2009.

\bibitem{Calcagno11}
C.~Calcagno, D.~Distefano, P.~O'Hearn, and H.~Yang, ``Compositional shape
  analysis by means of bi-abduction,'' \emph{JACM}, vol.~58, no.~6, 2011.

\bibitem{Li13}
B.~Li, I.~Dillig, T.~Dillig, K.~L. McMillan, and M.~Sagiv, ``Synthesis of
  circular compositional program proofs via abduction,'' in \emph{TACAS}, 2013.

\bibitem{Albarghouthi12b}
A.~Albarghouthi, R.~Kumar, A.~V. Nori, and S.~K. Rajamani, ``Parallelizing
  top-down interprocedural analyses,'' in \emph{PLDI}, 2012.

\bibitem{Larraz14}
D.~Larraz, K.~Nimkar, A.~Oliveras, E.~Rodr\'{\i}guez-Carbonell, and A.~Rubio,
  ``Proving non-termination using {Max-SMT},'' in \emph{CAV}, 2014.

\bibitem{HandbookOfSAT2009}
A.~Biere, M.~J.~H. Heule, H.~van Maaren, and T.~Walsh, Eds., \emph{Handbook of
  Satisfiability}.\hskip 1em plus 0.5em minus 0.4em\relax IOS Press, 2009.

\bibitem{Colonetal2003CAV}
M.~Col{\'o}n, S.~Sankaranarayanan, and H.~Sipma, ``{Linear Invariant Generation
  Using Non-linear Constraint Solving},'' in \emph{CAV}, 2003.

\bibitem{Bradley05a}
A.~R. Bradley, Z.~Manna, and H.~B. Sipma, ``Linear ranking with reachability,''
  in \emph{CAV}, 2005.

\bibitem{Schrijver}
A.~Schrijver, \emph{{Theory of Linear and Integer Programming}}.\hskip 1em plus
  0.5em minus 0.4em\relax Wiley, 1998.

\bibitem{Borrallerasetal2011JAR}
C.~Borralleras, S.~Lucas, A.~Oliveras, E.~Rodr\'{\i}guez-Carbonell, and
  A.~Rubio, ``{SAT Modulo Linear Arithmetic for Solving Polynomial
  Constraints},'' \emph{JAR}, vol.~48, no.~1, 2012.

\bibitem{JovanovicMoura2012IJCAR}
D.~Jovanovic and L.~M. de~Moura, ``Solving non-linear arithmetic,'' in
  \emph{IJCAR}, 2012.

\bibitem{Larraz13}
D.~Larraz, A.~Oliveras, E.~Rodr\'{\i}guez-Carbonell, and A.~Rubio, ``Proving
  termination of imperative programs using {Max-SMT},'' in \emph{FMCAD}, 2013.

\bibitem{Sharma11}
R.~Sharma, I.~Dillig, T.~Dillig, and A.~Aiken, ``Simplifying loop invariant
  generation using splitter predicates,'' in \emph{CAV}, 2011.

\bibitem{Ball01}
T.~Ball and S.~K. Rajamani, ``The \tool{SLAM} toolkit,'' in \emph{CAV}, 2001.

\bibitem{Henzinger03}
T.~A. Henzinger, R.~Jhala, R.~Majumdar, and G.~Sutre, ``Software verification
  with \tool{BLAST},'' in \emph{SPIN}, 2003.

\bibitem{Clarke05}
E.~M. Clarke, D.~Kroening, N.~Sharygina, and K.~Yorav, ``\tool{SATABS}:
  {SAT}-based predicate abstraction for {ANSI-C},'' in \emph{TACAS}, 2005.

\bibitem{McMillan06}
K.~L. McMillan, ``Lazy abstraction with interpolants,'' in \emph{CAV}, 2006.

\bibitem{Podelski07}
A.~Podelski and A.~Rybalchenko, ``\tool{ARMC}: The logical choice for software
  model checking with abstraction refinement,'' in \emph{PADL}, 2007.

\bibitem{Bradley11}
A.~R. Bradley, ``{SAT}-based model checking without unrolling,'' in
  \emph{VMCAI}, 2011.

\bibitem{Grebenshchikov12}
S.~Grebenshchikov, N.~P. Lopes, C.~Popeea, and A.~Rybalchenko, ``Synthesizing
  software verifiers from proof rules,'' in \emph{PLDI}, 2012.

\bibitem{Albarghouthi12}
A.~Albarghouthi, A.~Gurfinkel, and M.~Chechik, ``Whale: an interpolation-based
  algorithm for inter-procedural verification,'' in \emph{VMCAI}, 2012.

\bibitem{Cimatti12}
A.~Cimatti and A.~Griggio, ``Software model checking via {IC3},'' in
  \emph{CAV}, 2012.

\bibitem{Hoder12}
K.~Hoder and N.~Bj{\o}rner, ``Generalized property directed reachability,'' in
  \emph{SAT}, 2014.

\bibitem{DBLP:conf/popl/FlanaganQ02}
C.~Flanagan and S.~Qadeer, ``Predicate abstraction for software verification,''
  in \emph{POPL}, 2002.

\bibitem{DBLP:conf/sas/McMillan03}
K.~L. McMillan, ``Craig interpolation and reachability analysis,'' in
  \emph{SAS}, 2003.

\bibitem{Cousot77}
P.~Cousot and R.~Cousot, ``Abstract interpretation: {A} unified lattice model
  for static analysis of programs by construction or approximation of
  fixpoints,'' in \emph{POPL}, 1977.

\bibitem{Dillig13}
I.~Dillig, T.~Dillig, B.~Li, and K.~L. McMillan, ``Inductive invariant
  generation via abductive inference,'' in \emph{OOPSLA}, 2013.

\bibitem{Pasareanu04}
C.~S. Pasareanu and W.~Visser, ``Verification of \textsf{Java} programs using
  symbolic execution and invariant generation,'' in \emph{SPIN}, 2004.

\bibitem{DBLP:conf/popl/YorshYC08}
G.~Yorsh, E.~Yahav, and S.~Chandra, ``Generating precise and concise procedure
  summaries,'' in \emph{POPL}, 2008.

\bibitem{DBLP:conf/sat/LarrazORR14}
D.~Larraz, A.~Oliveras, E.~Rodr{\'{\i}}guez{-}Carbonell, and A.~Rubio,
  ``Minimal-model-guided approaches to solving polynomial constraints and
  extensions,'' in \emph{SAT}, 2014.

\bibitem{DBLP:conf/cav/BofillNORR08}
M.~Bofill, R.~Nieuwenhuis, A.~Oliveras, E.~Rodr{\'{\i}}guez{-}Carbonell, and
  A.~Rubio, ``The barcelogic {SMT} solver,'' in \emph{CAV}, 2008.

\bibitem{Gupta09}
A.~Gupta and A.~Rybalchenko, ``\tool{InvGen}: An efficient invariant
  generator,'' in \emph{CAV}, 2009.

\bibitem{DBLP:conf/pldi/GulwaniSV08}
S.~Gulwani, S.~Srivastava, and R.~Venkatesan, ``Program analysis as constraint
  solving,'' in \emph{PLDI}, 2008.

\bibitem{DBLP:conf/pldi/BeyerHMR07}
D.~Beyer, T.~A. Henzinger, R.~Majumdar, and A.~Rybalchenko, ``Path
  invariants,'' in \emph{PLDI}, 2007.

\bibitem{DBLP:conf/tacas/GulavaniCNR08}
B.~S. Gulavani, S.~Chakraborty, A.~V. Nori, and S.~K. Rajamani, ``Automatically
  refining abstract interpretations,'' in \emph{TACAS}, 2008.

\bibitem{DBLP:journals/fac/BradleyM08}
A.~R. Bradley and Z.~Manna, ``Property-directed incremental invariant
  generation,'' \emph{Formal Asp. Comput.}, vol.~20, no. 4-5, 2008.

\bibitem{Mine06}
A.~Min{\'e}, ``The octagon abstract domain,'' \emph{Higher-Order and Symbolic
  Computation}, vol.~19, no.~1, 2006.

\bibitem{DBLP:conf/tacas/JhalaM06}
R.~Jhala and K.~L. McMillan, ``A practical and complete approach to predicate
  refinement,'' in \emph{TACAS}, 2006.

\bibitem{DBLP:conf/cav/SharmaNA12}
R.~Sharma, A.~V. Nori, and A.~Aiken, ``Interpolants as classifiers,'' in
  \emph{CAV}, 2012.

\bibitem{DBLP:conf/sigsoft/GulavaniHKNR06}
B.~S. Gulavani, T.~A. Henzinger, Y.~Kannan, A.~V. Nori, and S.~K. Rajamani,
  ``{SYNERGY:} a new algorithm for property checking,'' in \emph{FSE}, 2006.

\bibitem{DBLP:conf/tacas/GulavaniR06}
B.~S. Gulavani and S.~K. Rajamani, ``Counterexample driven refinement for
  abstract interpretation,'' in \emph{TACAS}, 2006.

\bibitem{DBLP:conf/pldi/DilligDA12}
I.~Dillig, T.~Dillig, and A.~Aiken, ``Automated error diagnosis using abductive
  inference,'' in \emph{PLDI}, 2012.

\bibitem{DBLP:conf/tacas/000115}
D.~Beyer, ``Software verification and verifiable witnesses - (report on
  {SV-COMP} 2015),'' in \emph{TACAS}, 2015.

\bibitem{Beyer11}
D.~Beyer and M.~E. Keremoglu, ``{CPAchecker}: {A} tool for configurable
  software verification,'' in \emph{CAV}, 2011.

\bibitem{Kahsai15}
T.~Kahsai, J.~A. Navas, A.~Gurfinkel, and A.~Komuravelli, ``The {SeaHorn}
  verification framework,'' in \emph{CAV}, 2015.

\bibitem{Press02}
W.~H. Press, S.~A. Teukolsky, W.~T. Vetterling, and B.~P. Flannery,
  \emph{Numerical Recipes in C++ (2Nd Ed.): The Art of Scientific
  Computing}.\hskip 1em plus 0.5em minus 0.4em\relax New York, NY, USA:
  Cambridge University Press, 2002.

\end{thebibliography}


\appendices
\section{Proofs}
\label{sect:proofs}

\begin{proof}[Proof of \rT{thm:condSafe-sound}]
That $\QQ$ is a conditional inductive invariant follows directly from the
structure of the generated constraints, which are (modulo renaming) also
discussed in \cite{Bradley05a}.

We prove the claim about conditional safety by contradiction via induction over
the length of evaluations.
Assume that there is an unsafe execution
$$
  (\ell_1, \vect{v}_1) \to_{t_1} (\ell_2, \vect{v}_2) \to_{t_2} \ldots \to_{t_n} (\ell_{n}, \vect{v}_{n})
$$
of length $n \in \mathbb{N}_{> 1}$ such that $t_1 \in \entries_\cC \cup \cC$ (i.e.,
$\ell_2$ is always a location in $\cC$), $t_n = t_{\textrm{exit}}$, $\vect{v}_2 \models
\cQ(\ell_2)$ and $\vect{v}_{n} \not\models \varphi$. We will show that no
such evaluation can exist, implying our proposition as the special case $t_1 \in
\entries_{\cC}$.

As the component graph is a DAG, $t_1 \in \entries_{\cC} \cup \cC$ and
$t_\textrm{exit}$ is an exit transition of $\cC$, we have $t_i \in \cC$ for all $1 < i
< n$.

We first consider the case $n = 2$ ($n = 1$ would be the case where $t_1$ is
both an entry and exit transition, and thus infeasible).
Let $t_2 = (\ell_{1}, \tau_2, \ell_2)$.
Then, $\vect{v}_2 \models \QQ(\ell_2) \equiv \sigma(I_{\ell_2,k})$ by choice and
definition, and $\sigma(I_{\ell_2, k}) \land \tau_2 \limplies \varphi'$ by constraint
$\safety_k$. Thus, no unsafe evaluation of length $2$ is possible.

We now assume $n > 2$ and that the proposition has been proven for evaluations
of length $n - 1$.
Let $t_2 = (\ell_2, \tau_2, \ell_3)$.
For length $n$, we have that the valuations $\vect{v}_2$, $\vect{v}_3$ satisfy
$\tau_2$, and $\vect{v}_2 \models \cQ(\ell_2) \equiv \sigma(I_{\ell_2,k})$.
These are the premises of our consecution constraint
$\consecution_{t_2,k} \equiv I_{\ell_2,k} \land \tau_2 \limplies I_{\ell_3,k}$,
and thus $\vect{v}_3 \models \sigma(I_{\ell_3,k}) \equiv \cQ(\ell_3)$.
Hence, we instead have to consider the evaluation of length $n-1$ starting in
$(\ell_2, \vect{v}_2)$, which by our hypothesis is infeasible.
\end{proof}

\begin{proof}[Proof of \rT{thm:checkSafe-sound}]
We prove the proposition by induction over the number $u$ of recursive calls of
\checkSafe.

In the base case $u = 0$, we have that $\tau_{\mathrm{exit}} \limplies \varphi'$, i.e.,
the condition to prove is always a consequence of using the transition
$t_{\mathrm{exit}}$, and the claim trivially holds.

Let now $u > 0$, and we assume that the proposition has been shown for all calls
of \checkSafe that return \textsf{Safe} and need at most $u - 1$ recursive calls
of \checkSafe.

We now consider a program evaluation
$$
  (\ell_0, \vect{v}_0) \to_{t_0} \ldots \to_{t_{m-1}} (\ell_{m}, \vect{v}_{m}) \to_{t_m} \ldots \to_{t_n} (\ell_n, \vect{v}_n)
$$
with $t_{m - 1} \in \entries_{\cC}$ and $t_n = t_{\textrm{exit}}$.

First, we consider the case that \checkSafe returns \textsf{Safe} in line 11,
where all preconditions are satisfied for all entry transitions.
By the condition $\forall L \in \QQ(\ell_m).\; \textsf{res}[t_{m-1}, L] = \textsf{Safe}$ and our
induction hypothesis, we know that the program is safe for $(t_{m-1},
\QQ(\ell_m))$. By construction of $\QQ$ and \rT{thm:condSafe-sound}, this then
implies that the program is safe for $(t_\textrm{exit}, \varphi)$.

The second case is returning the result of \checkSafe on the narrowed
program in lines 19 and 20. For this, we need to prove that our program narrowing is
indeed correct.
Assume now that the considered evaluation is unsafe, i.e., that $\vect{v}_n
\not\models \varphi$.
We will show that our narrowing preserves unsafe evaluations. Then, as the
recursive call of \checkSafe (with recursion depth $u-1$) is correct by our
induction hypothesis, we can conclude that \textsf{Safe} is only returned if
there are no unsafe evaluations.

We first consider the narrowing $\hat{\entries}_\cC$. By our induction
hypothesis, we know that
 $\vect{v}_m \models (\bigwedge_{L \in \QQ(\ell_m), \textsf{res}[t_{m-1}, L] =
   \textsf{Safe}} L')$ holds.
Now assume that the narrowed version of $t_{m - 1}$ is not enabled anymore
because of the added condition.
Then $\vect{v}_m \not\models \lnot(\bigwedge_{L \in \cQ({\ell_m}),
  \textsf{res}[t_{m-1},L] = \textsf{Maybe}} L')$ holds, and thus
 $\vect{v}_m \models \bigwedge_{L \in \cQ(\ell_m)} L'$.
But then, our evaluation is safe (by the same argument as in the proof of
\rT{thm:condSafe-sound}), contradicting our assumption that the considered
evaluation is unsafe. Thus, unsafe evaluations are not broken by our narrowing
of entry transitions.

Similarly, we now consider the narrowing $\hat{\cC}$. We consider an
evaluation step $(\ell_w, \vect{v}_w) \to_{t_w} (\ell_{w+1}, \vect{v}_{w+1})$
with $t_w \in \hat\cC$.
If the narrowed version $\hat{t}_w \in \hat\cC$ of $t_w$ cannot be used, then
either $\vect{v}_w \models \QQ(\ell_w)$ or $\vect{v}_{w+1} \models
\QQ(\ell_{w+1})$ holds, and again, by an argument similar to the proof of 
\rT{thm:condSafe-sound}, this contradicts the assumption that our evaluation is
unsafe. Thus, unsafe evaluations are preserved by narrowing of the program
component.
\end{proof}

\section{Choice of Weights in the Max-SMT Formulation}
\label{sect:weights}
In VeryMax, all weights for the initiation conditions are currently
the same. Choosing different weights would bias the solver towards
satisfying certain initiation conditions.  Weights play a crucial role
in VeryMax to favor invariants that imply postconditions over ones
that disable transitions (cf. Sect.  IV-C-a). The chosen MaxSMT
framework also allows adding other techniques with lower weights,
e.g. the techniques from \cite{Larraz14} that strengthen single program
transitions during the analysis.


\end{document}